\documentclass[runningheads]{llncs}

\usepackage{amsmath}
\usepackage{amssymb}

\newcommand{\df}{\textit}

\newcommand{\B}{{B}}
\newcommand{\K}{\mathcal{C}}
\newcommand{\MyOmega}{\mathrm{\Omega}}
\newcommand{\MyTheta}{\mathrm{\Theta}}

\newcommand{\bx}{\mathbf{x}}
\newcommand{\by}{\mathbf{y}}
\newcommand{\bz}{\mathbf{z}}
\newcommand{\bu}{\mathbf{u}}
\newcommand{\bw}{\mathbf{w}}

\newcommand{\bconn}{\{\land, \lor, \oplus\}}
\newcommand{\bl}{B_l}

\newcommand{\no}{\overline}

\DeclareMathOperator{\lca}{lca}

\spnewtheorem*{corollary*}{Corollary}{\upshape\bfseries}{\itshape}
\spnewtheorem*{theorem*}{Theorem}{\upshape\bfseries}{\itshape}
\spnewtheorem*{remark*}{Remark}{\itshape}{\rmfamily}

\begin{document}

\title{Checking Tests for Read-Once Functions\\over Arbitrary Bases}

\author{Dmitry V. Chistikov}
\institute{%
Faculty of Computational Mathematics and Cybernetics\\
Moscow State University, Russia\\
\email{dch@cs.msu.ru}}

\authorrunning{\quad}
\titlerunning{\quad}

\maketitle

\begin{abstract}
A Boolean function is called read-once over a basis $B$
if it can be expressed by a formula over $B$ where no variable
appears more than once.
A checking test for a read-once function $f$ over $B$
depending on all its variables
is a set of input vectors distinguishing $f$ from all
other read-once functions of the same variables.
We show that every read-once function $f$ over $B$ has
a checking test containing $O(n^l)$ vectors, where $n$
is the number of relevant variables of $f$ and $l$ is
the largest arity of functions in $B$.
For some functions, this bound cannot be improved
by more than a constant factor.
The employed technique involves reconstructing $f$ from
its $l$-variable projections and provides a stronger form
of Kuznetsov's classic theorem on read-once representations.

\keywords{read-once Boolean function, checking test, complexity,
teaching dimension, equivalence query, membership query}.
\end{abstract}

\section{Introduction}
\label{s:intro}

Let $\B$ be an arbitrary set of Boolean functions. A function $f$ is called
\df{read-once} over $\B$ iff it can be expressed by a formula over $\B$ where
no variable appears more than once.

Let $f(x_1, \ldots, x_n)$ be a read-once function over $\B$ that
depends on all its variables. Then a set $M$ of $n$-bit vectors
is called a \df{checking test} for $f$ iff for any other read-once function $g(x_1, \ldots, x_n)$
over $\B$ there exists a vector $\alpha \in M$ such that $f(\alpha) \ne g(\alpha)$.
In other words, $M$ is a checking test for $f$ iff the values of $f$ on vectors from $M$ distinguish
$f$ from all other read-once functions $g$ of the same variables. Note that all
these \df{alternatives} $g$, unlike the \df{target function} $f$, may have irrelevant variables.

Denote by $\bl$ the basis of all $l$-variable Boolean functions.
The goal of this paper is to prove that all $n$-variable read-once
functions over $\bl$ have checking tests containing $O(n^l)$ vectors.
More generally, for an arbitrary basis $\B$ and a read-once function $f$ over $\B$, denote
by $T_\B(f)$ the smallest possible number of vectors in a checking test for $f$.
This value can be regarded as the \df{checking test complexity} of $f$. We show that
\begin{equation*}
    T_{\bl}\bigl(f(x_1, \ldots, x_n)\bigr) \le 2^l \cdot \binom{n}{l}
\end{equation*}
and, therefore, for any finite basis $\B$ and any sequence of read-once
functions $f_n$ of $n$ variables over $\B$, it holds that $T_{\B}(f_n) = O(n^l)$
as $n \to \infty$, where $l$ is the largest arity of functions from $\B$.

This result is based on the previously known relevance hypercube method
by Voronenko~\cite{aavpmi23}. Our main contribution is the proof that the method
is \df{correct} for all bases $\bl$ for an arbitrary $l$, i.\,e., it provides
a way to construct checking tests of specified \df{length} (cardinality) for
all read-once functions over these bases. Previous results give proofs only
for~$l \le 5$~\cite{aavmvk11,aavpmi23,vchb5nsk}.

It should be pointed out that the bound $T_{\bl}(f) = O(n^{l + 1})$ can
be extracted from the related paper on the exact identification problem by
Bshouty, Hancock and Hellerstein~\cite{bhhgen}. Our result has the following
advantages. Firstly, for some functions the bound $O(n^l)$ cannot be improved
by more than a constant factor (it matches the known lower bound
$\MyOmega(n^l)$ for $n$-ary disjunction up to a constant factor).
Secondly, checking tests constructed by the relevance hypercube method have
regular structure. In short, we show that every read-once function can
be unambiguously reconstructed from a set of its $l$-variable projections with
certain properties. This fact may look natural at first sight, but turns out
a tricky thing to prove after all.

\section{Background and Related Work}

Our result has some interesting consequences related to computational learning
theory. For instance, it is known that checking tests can be used to implement
\emph{equivalence queries} from Angluin's learning model~\cite{a88}.
For the problem of identifying an unknown read-once function over an arbitrary
finite basis $B$ with queries, it turns out that non-standard \df{subcube
identity queries} can be efficiently used~\cite{dchid}. A subcube identity query basically
asks whether a specified projection of the unknown function $f$ is constant,
i.\,e., whether a given partial assignment of constants to input variables unambiguously
determines the value of $f$.

It follows from our results that for any finite basis $B$,
the problem of learning an unknown read-once function over $B$ can be solved
by an algorithm making $O(n^{l + 2})$ membership and subcube identity queries,
which is polynomial in $n$ (here $l$ is the largest arity of functions in $B$
and a standard \df{membership query} is simply a request for the value of $f$
on a given input vector). This result builds upon an algorithm by Bshouty,
Hancock and Hellerstein~\cite{bhhgen}, which is a strong generalization
of a classic exact identification algorithm by Angluin, Hellerstein and Karpinski~\cite{ahk}.

Closely related to the notion of checking test complexity is the definition
of teaching dimension introduced by Goldman and Kearns~\cite{gkteach}.
A \df{teaching sequence} for a Boolean concept (a Boolean function $f$)
in a known class $\K$ is a sequence of labeled instances (pairs of the form
$\langle \alpha, f(\alpha) \rangle$) consistent with only one function $f$ from $\K$.
The \df{teaching dimension} of a class is the smallest number $t$ such that
all concepts in the class have teaching sequences of length at most $t$.

In test theory, which dates back to 1950s~\cite{chya}, the corresponding definition
is that of the \df{Shannon function for test complexity}, which is the largest
test complexity of an $n$-variable function. In these terms, our Corollary
can be restated as follows: $T_{\bl}(n) = O(n^l)$, where $T_{\bl}(n)$ is
the Shannon function for checking test complexity of read-once Boolean functions
(i.\,e., the maximum of $T_{\bl}(f)$ over all $n$-variable read-once functions over
$\bl$). It must be stressed that in the definition of a checking test used
in this paper, the target function is required to depend on all
its variables. (One needs to test all $2^n$ input vectors to distinguish
the Boolean constant $0$ from all read-once conjunctions of $n$ literals.)

Another appealing problem is that of obtaining bounds on
the value of $T_\B (f)$ for individual read-once functions $f$.
The bound
$T_{\bl}(x_1 \lor \ldots \lor x_n) = \MyTheta(n^l)$
is obtained in~\cite{aavpmi23} and generalized in~\cite{aavdisj}.
In~\cite{dchsub},
it is shown that for a wide class of bases including $\bl$, $l \ge 2$,
there exist pairs of read-once functions $f, f'$ such that $f'$ is obtained
from $f$ by substituting a constant for a variable and $T_\B(f') > T_\B(f)$.
This result shows that lower bounds on $T_\B(f)$ cannot generally be obtained
by simply finding projections of $f$ that are already known to require a large
number of vectors in their checking tests.

For the basis $B_2$, individual bounds on the checking test complexity are
obtained in~\cite{vchindkzn}. In~\cite{zcvaaco}, it is shown that almost
all read-once functions over the basis $\{\lor, \oplus\}$ have a relatively
small checking test complexity of $O(n \log n)$, as compared to the maximum
of $\MyTheta(n^2)$ (even if alternatives are arbitrary read-once functions
over $B_2$ and not necessarily read-once over $\{\lor, \oplus\}$).
For the standard basis $\{\land, \lor, \neg\}$, it is known that
$n + 1 \le T_{\{\land, \lor, \neg\}}(n) \le 2 n + 1$~\cite{dchaouu}, and individual
bounds can be deduced from those for the monotone basis
$\{\land, \lor\}$~\cite{bvcao3,dchctmrof}.

\section{Basic Definitions}
\label{s:def}

A variable $x_i$ of a Boolean function $f(x_1, \ldots, x_n)$ is called
\df{relevant} (or essential) if there exist two $n$-bit vectors $\alpha$ and $\beta$
differing only in the $i$th component such that $f(\alpha) \ne f(\beta)$.
If $x_i$ is relevant to $f$, then $f$ is said to \df{depend} on $x_i$.

In this paper, we call a pair of functions $f(x_1, \ldots, x_n)$
and $g(y_1, \ldots, y_n)$ \df{similar} if for some constants
$\sigma, \sigma_1, \ldots, \sigma_n \in \{0, 1\}$ and for some permutation $\pi$
of $\{1, \ldots, n\}$ the following equality holds:
\begin{equation*}
    f(x_1, \ldots, x_n) \equiv g^{\sigma_{\vphantom{1}}}_{\vphantom{(}}
    \bigl(x_{\pi(1)}^{\sigma_1}, \ldots, x_{\pi(n)}^{\sigma_n}\bigr),
\end{equation*}
where $z^\tau$ stands for $z$ if $\tau = 1$ and for $\no z$ if $\tau = 0$.
A Boolean function $f(x_1, \ldots, x_n)$, $n \ge 3$, is called \df{prime} if
it has no decomposition of the form
\begin{equation*}
    f(x_1, \ldots, x_n) \equiv
    g\!\left(\,
    h(x_{\pi(1)}, \ldots, x_{\pi(k)}),\,
    x_{\pi(k + 1)}, \ldots, x_{\pi(n)}\right),
\end{equation*}
where $1 < k < n$ and $\pi$ is a permutation of $\{1, \ldots, n\}$.

The structure of formulae expressing read-once functions can be represented
by rooted trees. A \df{tree} of a read-once function $f(x_1, \ldots, x_n)$
over $\bl$ has $n$ leaves labeled with literals of different variables and one or more
internal nodes labeled with functions from $\bl$ and symbols $\circ \in \bconn$
of arbitrary arity (possibly exceeding $l$). We assume without loss
of generality that such trees also have the following properties:
\begin{enumerate}
\renewcommand{\labelenumi}{\theenumi)}
\item any internal node is labeled either with a prime function or with
    a symbol $\circ \in \bconn$;
\item internal nodes labeled with identical symbols $\circ \in \bconn$
    are not adjacent.
\end{enumerate}
One can readily see that every read-once function over $\bl$ has at least
one tree of this form.

In the sequel, variables are usually identified with corresponding leaves in the tree.
Denote by $\lca(y_1, \ldots, y_m)$ the \df{least common ancestor} of variables $y_1, \ldots, y_m$,
i.\,e., the last common node in (simple) paths from the root of the tree to $y_1, \ldots, y_m$.

Suppose that $T$ is a tree of a read-once function and $v$ is its internal
node. By $T_v$ we denote the \df{subtree of $T$ rooted at $v$}, i.\,e.,
the rooted tree that has root $v$ and contains all descendants of $v$.
If $w_1, \ldots, w_p$ are \df{children} (direct descendants) of $v$, then
subtrees $T_{w_1}, \ldots, T_{w_p}$ are called \df{subtrees of the node} $v$.
If $x$ is a leaf of $T$ contained in $T_v$, then by $T_v^x$ we denote
a (unique) subtree $T_{w_j}$ containing $x$.
Finally, subtrees of the root node of a tree are called \df{root subtrees}.

\section{The Relevance Hypercube Method}
\label{s:rhc}

This section is devoted to the review of the relevance hypercube method
proposed by Voronenko in~\cite{aavpmi23}. This method has been known to be correct
for the bases $\bl$ if $l \le 5$ (see~\cite{aavpmi23,vchb5nsk}).

From now on, we will use the term ``read-once function'' instead of ``read-once
function over $\bl$''.
We use boldface letters to denote vectors
(often treated as sets) of variables.

Let $f$ be a read-once function depending on variables $\bx = \{x_1, \ldots, x_n\}$.
Suppose that $H$ is a set of $2^l$ input vectors disagreeing at most in
$i_1$th, \ldots, $i_l$th components such that the restriction of $f$ to $H$ (which is an
$l$-variable Boolean function) depends on all its $l$ variables
$\bx' = \{ x_{i_1}, \ldots, x_{i_l} \}$.
Then $H$ is called a \df{relevance hypercube} (or an essentiality hypercube)
of dimension $l$ for these variables $\bx'$.
Any relevance hypercube can be identified with a partial assignment $p$ of constants to input
variables such that the induced projection $f_p$ depends on all its
$l$ variables. Such assignments are called $l$-justifying in~\cite{bhhgen}.

\begin{remark*}
For some functions $f$ and some subsets of their variables relevance
hypercubes do not exist. For instance, one may easily check that the function
$d(x, u_0, u_1) = (\no x \land u_0) \lor (x \land u_1)$
has no relevance hypercubes for the set $\bu = \{ u_0, u_1 \}$.
As indicated below, the absence of relevance hypercubes is a major obstacle to proving
the correctness of the relevance hypercube method (see also~\cite{aavpmi23,vchb5nsk}).
At the same time, for some functions there exist subsets of variables with
more than one relevance hypercube. An example is given by the same
function $d$ and the set $\bu' = \{ x, u_0 \}$.
\end{remark*}

Any set $M$ of input vectors is called a \df{relevance hypercube set}
of dimension $l$ for $f$ if it contains a relevance hypercube $H$ for every
$l$-sized subset of $\bx$, for which such a hypercube exists.
(Recall that $\bx$ is the set of all variables relevant to $f$.)
In other words: consider all $l$-sized subsets $\bw \subseteq \bx$
such that $f$ has a relevance hypercube for $\bw$. A set $M$ is
a relevance hypercube set iff $M$ contains
at least one relevance hypercube for each subset $\bw$ of this kind.
It was conjectured that any such set is
a checking test for $f$.

Suppose that $f$ is a read-once function that depends
on $n$ variables $\bx$. Construct a \df{relevance table} with
$\binom{n}{l}$ rows and two columns by the following rule. First, fill the
first cells of all rows with different $l$-sized subsets of $\bx$.
Then for each row, if the first cell contains a
subset $\bw \subseteq \bx$, put in the second cell any
relevance hypercube for $\bw$ (along with the corresponding values of $f$)
if such a hypercube exists, or the symbol $*$ otherwise.

In~\cite{aavpmi23}, it is shown that any (valid) relevance table uniquely
determines a read-once function in the following sense.
Suppose that one knows that a function $g$ is read-once and agrees with
$f$ on all relevance hypercubes from a relevance
table $E$ for $f$. If one also knows that for each $*$-row in $E$ the function
$g$ does not have a relevance hypercube, then one can recursively reconstruct the \df{skeleton}
of $g$ (which is a tree $T'$ such that negating
some of its nodes' labels yields a correct tree representing $g$).
After that, the values of $f$ on the vectors from $E$ allow one to prove
that $g$ is equal to $f$.

It follows that a relevance hypercube set $M$ of dimension $l$ for
a read-once function $f$ is indeed a checking test for $f$ if $E$ contains
no $*$-rows. If for some $l$-sized subset of variables $\bw$ no relevance
hypercube exists, then a more sophisticated technique is needed to prove that
$f$ can still be reconstructed from its values on vectors from $M$.
The approach used in this paper is outlined in the following Section~\ref{s:preproof}.

\section{Assumptions and Notation}
\label{s:preproof}

In this section we make preliminary assumptions and introduce some notation.
All subsequent work, including the proof of our main theorem, is based on
the material presented here.

We start with an arbitrary read-once function $f$ over $\bl$, where $l \ge 3$.
Let $\bx$ be the set of variables relevant to $f$.
Suppose that $M$ is a relevance hypercube set of dimension $l$ for $f$.
Our goal is to prove that $M$ is a checking test for $f$, i.\,e.,
for any other read-once function $g(\bx)$ there exists
a vector $\alpha \in M$ such that $f(\alpha) \ne g(\alpha)$.

Take any read-once function $g(\bx)$ that agrees with $f(\bx)$ on all vectors
from $M$. Firstly and most importantly, we need to prove that root nodes of
these two functions' trees are labeled with similar functions.
If either of the root nodes is labeled with a symbol $\circ \in \bconn$,
then this can be done with the aid of techniques similar to those from~\cite{aavpmi23}.
Here we focus on the prime case, i.\,e, we assume that
\begin{gather*}
    f(\bx) = f^0(f_1(\bx^1), \ldots, f_s(\bx^s)), \\
    g(\by) = g^0(g_1(\by^1), \ldots, g_r(\by^r)),
\end{gather*}
where both $f^0$ and $g^0$ are prime, and $\bx^1 \cup \ldots \cup \bx^s$ and
$\by^1 \cup \ldots \cup \by^r$ are partitions of $\bx = \by$.
(Technically, we must first assume that $\by \subseteq \bx$, but it is easily
shown that no variable from $\bx$ can be irrelevant to $g$; see, e.\,g.,
Proposition~\ref{p:hcexpand} in the next section.)
Note that here $s \le l$ and $r \le l$.

Suppose we have already proved that $f^0$ and $g^0$ are similar.
As our second step, we need to show that partitions of input variables into
subtrees are identical in the representations above. In other words, we need to
show that each $\by^k$ is equal to some $\bx^i$.

These two steps, especially the first one, constitute the main
difficulties in proving the correctness of the method. The remaining part is
technical and can be done with the aid of induction on the depth of the tree
representing $f$. A short explanation of how this part is done is given at the
end of our main theorem's proof in Section~\ref{s:th}.

In the following sections, we will need the \df{colouring} of input variables $\bx$ defined by the following rule.
To each variable $x \in \bx$ we assign a (unique) colour $k \in \{1, \ldots, r\}$ such that $x \in \by^k$.
This definition provides a convenient way of relating functions $f$ and $g$
(i.\,e., their tree structure) to each other.

\section{Some Observations}

In this section we present three facts needed for the sequel.
A key observation is given by the following proposition.

\begin{proposition}
\label{p:lcag}
    Suppose that $g'$ is a projection of $g$ that depends
    on variables $x$ and $y$ having the same colour $k$.
    Also suppose that $\lca(x, y) = v$ in a tree $T'$ of $g'$.
    Then, if $v$ is labeled with a prime function $h$, it follows
    that all leaves in the subtree $(T')_v$ have the same colour $k$.
    Otherwise, if $v$ is labeled with a symbol $\circ \in \bconn$,
    it follows that all leaves in subtrees $(T')_v^x$ and $(T')_v^y$ have
    the same colour $k$.
\end{proposition}

\noindent
Proposition~\ref{p:lcag} follows from a simple fact that substitutions
of constants for variables of $g$ can result in removing nodes and subtrees
from $T'$, or in replacing nodes with trees that represent projections
of prime functions. Adjacent nodes labeled
with identical symbols $\circ \in \bconn$ are subsequently glued together,
but least common ancestors of each $\by^i$ either remain roots of single-coloured
subtrees, or ``support'' subsets of single-coloured subtrees of internal
nodes labeled with $\circ \in \bconn$.

For technical reasons, we will also need the following proposition,
which holds true for all discrete functions (not necessarily read-once
or even Boolean) and follows from Theorem~B in~\cite{davies}.

\begin{proposition}
\label{p:hcexpand}
    Suppose that $f$ is an arbitrary function depending
    on $n$ variables $\bx$. Also suppose that there
    exists a relevance hypercube for some $p$ variables $\bu \subseteq \bx$.
    Then for every $q$ such that $p \le q \le n$ there exists a relevance
    hypercube for some $q$-sized set of variables $\bw$ such that
    $\bu \subseteq \bw \subseteq \bx$.
\end{proposition}

\noindent
Last but not least, we will use the following fact (see, e.\,g.,~\cite{aavpmi23}).

\begin{proposition}
\label{p:eq}
    Suppose that a read-once function $f$ is represented by a tree
    and $p$ variables $\bu$ are
    taken from $p$ different subtrees of an internal node $v$,
    which is labeled with a prime function of arity $p$
    or with a symbol $\circ \in \bconn$. Then $f$
    has at least one relevance hypercube for $\bu$ and, moreover, restrictions
    of $f$ to all such hypercubes are:
    \vspace{-1ex}
    \begin{enumerate}
    \renewcommand{\theenumi}{\alph{enumi}}
    \renewcommand{\labelenumi}{\textup{(}\theenumi\textup{)}}
    \item similar to $h(z_1, \ldots, z_p)$ if $v$ is labeled with a prime
          function $h$;
    \item similar to $z_1 \circ \ldots \circ z_p$ if $v$ is labeled with
          a symbol $\circ \in \bconn$.
    \end{enumerate}
\end{proposition}

\section{Auxiliary Lemmas}
\label{s:lemmas}

Suppose that read-once functions $f$ and $g$ satisfy all the
assumptions made in Section~\ref{s:preproof} and $T$ is a tree of $f$.
Recall that by $\bx$ we denote the set of
all variables relevant to $f$.
We say that a set $\bu \subseteq \bx$ is \df{stable}
iff for any set $\bw$ such that
$\bu \subseteq \bw \subseteq \bx$ and any relevance
hypercube $H$ for $\bw$ there exists a relevance hypercube $H'$
for $\bu$ such that $H' \subseteq H$.

\begin{remark*}
    The definition of a stable set $\bu$ does not require the existence of
    relevance hypercubes for all sets $\bw$ such that $\bu \subseteq \bw \subseteq \bx$.
    What is says is that \emph{if} there exists such a relevance hypercube $H$,
    then there exists a relevance hypercube for $\bu$
    which is a subcube of $H$. For $\bw = \bx$, however, the definition requires
    that at least one relevance hypercube for $\bu$ exists.
\end{remark*}

\noindent
One can readily observe that all singleton subsets of $\bx$
are stable. Examples of sets that are not stable are given by $\bu = \{ u_0, u_1 \}$
(as witnessed by $\bw = \bu \cup \{ y \}$) for functions
$f_1 = (\no x \land d(y, u_0, u_1)) \lor (x \land (y \lor u_0 \lor u_1))$
and
$f_2 = (\no x \land d(y, u_0, u_1)) \lor (x \land (u_0 \lor u_1))$,
where
$d(y, u_0, u_1) = (\no y \land u_0) \lor (y \land u_1)$.
Our main ingredient in the proof of the main theorem is given by the following
lemma, which allows us to establish a link between the tree structure
of our functions $f$ and $g$.

\begin{lemma}
\label{l:mainl}
    For any stable set $\bu$ of at most $l$ variables of the function $f$,
    the function $g$ agrees with $f$ on some relevance hypercube for $\bu$.
\end{lemma}

\begin{proof}
We start with the definition of a stable set. First choose $\bw = \bx$
and conclude that $f$ has a relevance hypercube for $\bu$. By Proposition~\ref{p:hcexpand},
$f$ also has a relevance hypercube for some $l$-sized set of variables $\bw'$
such that $\bu \subseteq \bw'$.
It follows that $M$ contains some relevance hypercube for $\bw'$. Since
$f$ and $g$ agree on all vectors from $M$, and $\bu$ is stable, it also
follows that $f$ and $g$ agree on some relevance hypercube for $\bu$.
This concludes the proof.
\end{proof}

\noindent
More than once we will need subsets of input variables having specific structure.
We call a set $\bu$ \df{conservative}
iff for each internal node $v$ of $T$ labeled with a prime function $h$
the number of subtrees of $v$ containing at least one variable from $\bu$
is equal either to $0$, or to $1$, or to the arity of $h$.
To understand the intuition behind this term, consider the restriction of $f$
to any relevance hypercube for such a set. In the tree of such a restriction,
each node of $T$ labeled with a prime function is either preserved or discarded,
i.\,e., no constant substitutions and further transformations occur at these nodes.

\begin{lemma}
\label{l:st}
    All conservative sets are stable.
\end{lemma}

\begin{proof}
Let $\bu$ be a conservative set of variables.
We say that an internal node of $T$ is a \df{branching node} for $\bu$
iff at least two subtrees of $v$ contain leaves from $\bu$.
The proof is by induction over the number $b$ of branching nodes for $\bu$
in $T$. If $b = 0$, then $|\bu| \le 1$ and there is nothing to prove.
Suppose that $b \ge 1$ and $v = \lca(\bu)$.
Then $v$ is a branching node and all other branching nodes
are descendants of $v$. Therefore, $f$ can be expressed by a formula
\begin{equation*}
    h_0 \bigl(\, \bz^0,\  h( h_1(\bz^1), \ldots, h_m(\bz^m)) \,\bigr),
\end{equation*}
where $h$ is the function corresponding to the label of $v$,
sets $\bz^i$ and $\bz^j$ are disjoint for $i \ne j$
and (by our definition of a conservative set) variables $\bu$ are contained in each
$\bz^i$, $i \ge 1$, but not in $\bz^0$.
\par
Identify a relevance hypercube $H$ for some variables $\bw$ (here $\bu \subseteq \bw$) with a partial
assignment $p$ of constants to variables $\bx$. Split $p$ into $p_0, p_1, \ldots, p_m$
according to the partitioning given by $\bz^0, \bz^1, \ldots, \bz^m$.
Subsets of $\bu$ contained in $\bz^i$, $i \ge 1$, are conservative for trees
representing functions $h_i(\bz^i)$, and the number of branching nodes in any such tree
is at most $b - 1$.
By the inductive assumption, there exist partial assignments $p'_1, \ldots, p'_m$
which are extensions of $p_1, \ldots, p_m$ and restrict relevance hypercubes
for these subsets. Projections $h'_i$ induced by $p'_i$ depend
on these subsets of $\bu$. If we now choose an extension $p'_0$ of $p_0$
taking $h_0(\bz^0, u)$ to a literal of $u$,
the composition of $p'_0, p'_1, \ldots, p'_m$ will restrict the needed
relevance hypercube $H'$. This concludes the proof.
\end{proof}

\noindent
In Section~\ref{s:preproof}, we defined the colouring of variables induced by the read-once
representation of $g$.
The following lemmas reveal some properties of this colouring
that are related to the structure of $T$.
These properties reflect the observation formulated in
Proposition~\ref{p:lcag}.

\begin{lemma}
\label{l:j}
    Suppose that variables $x$ and $y$ both have colour $k$.
    Also suppose that the node $v = \lca(x, y)$ in $T$
    is labeled with a prime function $h$.
    Then all the leaves in the subtree $T_v$ have the same colour $k$.
\end{lemma}

\begin{proof}
Let $m$ be the arity of $h$. Take arbitrary variables $z_1, \ldots, z_{m - 2}$
such that $x, y, z_1, \ldots, z_{m - 2}$ are contained in $m$ different subtrees of $v$.
The set $\bu = \{x, y, z_1, \ldots, z_{m - 2}\}$ is conservative and,
therefore, stable (by Lemma~\ref{l:st}).
Since $h$ is prime and $f$ is read-once over $\bl$, we see that $m \le l$.
It follows from Lemma~\ref{l:mainl} that $f$ agrees with $g$ on some
relevance hypercube for $\bu$. By Proposition~\ref{p:eq}, the restriction
of $f$ to any such hypercube is similar to $h$.
Therefore, some projection $g'$ of $g$ is represented
by a tree $T'$ with exactly one internal node, which is labeled with a prime function.
Variables $x$ and $y$ are relevant to $g'$ and have the same colour $k$.
It then follows from Proposition~\ref{p:lcag} that all variables
$z_1, \ldots, z_{m - 2}$ have colour $k$ too. Since $z_1, \ldots, z_{m - 2}$
were chosen arbitrarily from their subtrees, we obtain that all leaves
in $T_v^{z_1}, \ldots, T_v^{z_{m - 2}}$ have colour $k$.
Repeating the same reasoning for initial pairs $x, z_1$ and $y, z_1$ in place of $x, y$ reveals
that all leaves in $T_v^y$ and $T_v^x$ also have the same colour $k$.
This concludes the proof.
\end{proof}

\begin{lemma}
\label{l:cs}
    For each colour $k \in \{1, \ldots, r\}$, there exists a unique
    index $i \in \{1, \ldots, s\}$ such that $\by^k \subseteq \bx^i$, i.\,e.,
    all variables coloured with $k$ belong to the set~$\bx^i$.
\end{lemma}

\begin{proof}
Take any two variables $x$ and $y$ having colour $k$. If they do not belong to
the same $\bx^i$, then they belong to different root subtrees of $T$.
Therefore, the node $\lca(x, y)$ is labeled with a prime
function $f^0$. By Lemma~\ref{l:j}, all leaves of $T$ have the same
colour, which is a contradiction.
\end{proof}

\noindent
Another way to state Lemma~\ref{l:cs} is to say that the partition $\by^1 \cup \ldots \cup \by^r$
is a refinement of $\bx^1 \cup \ldots \cup \bx^s$.

\begin{lemma}
\label{l:jr}
    For any non-root internal node $v$ in $T$ labeled with a prime function $h$
    of arity $r$ or greater, all leaves of $T_v$ have the same colour.
\end{lemma}

\begin{proof}
Let $v$ be an internal node of $T$ labeled with a prime function $h$
of arity at least $r$. If not all leaves of $T_v$ have the same colour, then
by Lemma~\ref{l:j} any two leaves from different subtrees of $v$
have different colours. It then follows that leaves of $T_v$ are coloured
with at least $r$ colours. By Lemma~\ref{l:cs}, leaves of other root subtrees
of $T$ cannot be coloured, since all colours are taken from the set
$\{1, \ldots, r\}$. This contradiction concludes the proof.
\end{proof}

\begin{lemma}
\label{l:j2}
    Suppose that variables $x$ and $y$ both have colour $k$.
    Also suppose that in $T$ the node $v = \lca(x, y)$
    is labeled with a symbol $\circ \in \bconn$.
    Then all the leaves in $T_v^x$ and $T_v^y$ have the same colour $k$.
\end{lemma}

\begin{proof}
Define the \df{depth} of a subtree as the maximum number of edges
on (shortest) paths from its root to its leaves. Let $d$ be the depth of $T_v$.
The proof is by induction over $d$. For $d = 1$, there is nothing to prove.
Suppose that $d \ge 2$ and $T_v^y$ contains a leaf $z$ having colour $k' \ne k$.
We claim that for some $m \ge 0$ there exist variables $z_0, z_1, \ldots, z_m$
such that the set $\bu = \{x, y, z_0, z_1, \ldots, z_m\}$ is conservative,
has cardinality at most $l$, and the restriction of $f$ to any relevance
hypercube for $\bu$ can be obtained by negating the inputs and/or the output
of some function
$
    x \circ f'(y, z_0, z_1, \ldots, z_m)
$,
where the colours of $y$ and $z_0$ are different and
$f'$ is either a prime function or a binary function ($m = 0$) from $\bconn$
different from $\circ$.
\par
First suppose that the root $w$ of $T_v^y$ is labeled with a prime function $h$.
Since not all leaves of $T_v^y$ have the same colour, it follows from
Lemma~\ref{l:j} that colours of leaves taken from different subtrees of $T_v^y$ are
different. Take arbitrary variables $z_0, z_1, \ldots, z_m$, $m \ge 1$, from all subtrees
except $T_w^y$ (one variable from each subtree). Now $y$ and $z_0$ have
different colours and $m + 2 \le r - 1$ by Lemma~\ref{l:jr}.
One can easily see that the set $\bu$ constructed
in this way is conservative by definition and has cardinality at most $l$, because $r \le l$.
Proposition~\ref{p:eq} then reveals that
the restriction of $f$ to any relevance hypercube for $\bu$ indeed has the
needed form.
\par
Now consider the case when $w$ is labeled with a symbol $\star \in \bconn$.
By definition of a tree representing a read-once function, $\star$ is different
from $\circ$. Observe that the depth of the subtree $T_w$ is less than or equal to $d - 1$,
so we can use the inductive assumption for $T_w$. If $y$ and $z$ belong to the
same subtree $T'$ of $w$, then it follows that only leaves from $T'$ can have the same
colour as $y$. In this case, any leaf from any other subtree can be chosen to be $z_0$.
In the other case, if $y$ and $z$ belong to different subtrees, simply put $z_0 = z$.
One can now see that $m = 0$ and $\bu = \{ x, y, z_0 \}$ satisfy all the stated
conditions.
\par
Now apply Lemma~\ref{l:mainl} to the set $\bu$ (recall that all conservative sets
are stable by Lemma~\ref{l:st}). It follows that $g$ agrees with $f$ on some
relevance hypercube for $\bu$. By our choice of $\bu$, this means that
$g$ has a projection $g'$ of the form specified above. In the tree of $g'$,
the root is adjacent to the leaf labeled with a literal of $x$ and to the other
internal node, whose children are $y, z_0, z_1, \ldots, z_m$. Since $x$ and $y$
have the same colour, it follows from Proposition~\ref{p:lcag} that $z_0$
has the same colour as $x$, which contradicts our choice of $z_0$.
This completes the proof.
\end{proof}

\section{Main Theorem}
\label{s:th}

\begin{theorem*}
    For any read-once function $f$ over $\bl$, $l \ge 3$, depending
    on all its variables, any relevance hypercube set of dimension $l$ for $f$
    constitutes a checking test for $f$.
\end{theorem*}

\begin{proof}
Let $M$ be a relevance hypercube set of dimension $l$ for $f$.
By $g$ denote an alternative read-once function that agrees with $f$
on all vectors from $M$. Suppose that $f$ and $g$ satisfy all the
assumptions made in Section~\ref{s:preproof} and let $T$ be a tree of $f$.
Choose a subset $\bu$ of $T$'s leaves according to the following (non-de\-ter\-mi\-nis\-tic)
recursive rules:
\begin{enumerate}
\item Put $\bu = \bu(v_0)$, where $v_0$ is the root of $T$.
\item If all leaves in $T_v$ have the same colour, then
    $\bu(v) = \{x_i\}$ for some leaf $x_i$ contained in $T_v$.
\item Otherwise:
    \begin{enumerate}
    \item if $v$ is labeled with a prime function $h$, then
        $\bu(v) = \bigcup \bu(v_i)$ over all children $v_i$ of $v$;
    \item otherwise, if $v$ is labeled with a symbol $\circ \in \bconn$, then
        $\bu(v) = \bu'(v) \cup \bu''(v)$, where $\bu'(v) = \bigcup \bu(v_i)$ over all
        multi-coloured subtrees $T_{v_i}$ of $v$, and $\bu''(v) = \bigcup \bu(v_j)$ over
        some subset of all single-coloured subtrees $T_{v_j}$ of $v$ that contains
        one subtree of each colour.
    \end{enumerate}
\end{enumerate}
One can easily see that $\bu$ is conservative and, by Lemma~\ref{l:st}, stable.
By Lemmas~\ref{l:j} and \ref{l:j2}, it contains exactly
$r$ leaves.
It follows from Lemma~\ref{l:mainl} that $g$ agrees with $f$
on some relevance hypercube $H$ for $\bu$.
Since all elements of $\bu$ have different colours, it follows from Proposition~\ref{p:eq} that
the restriction of $g$ to $H$ is similar to $g^0$. On the other
hand, $\bu$ contains at least one leaf from each root subtree of $T$,
so the restriction of $f$ to $H$ has the form
\begin{equation*}
    f' = f^0(f'_1, \ldots, f'_s),
\end{equation*}
where functions $f'_i$ depend on disjoint sets of variables from $\bu$.
These two restrictions are equal, so $f'$ is a prime function,
which is only possible if $r = s$ and $g^0$ is similar to $f^0$.
By Lemma~\ref{l:cs}, sets of root subtrees' variables are the same
for $f$ and $g$.
This means that
\begin{equation*}
    g(\bx) = f^0(g'_1(\bx^1), \ldots, g'_s(\bx^s)).
\end{equation*}
Since a relevance hypercube set for $f$ contains relevance
hypercube sets for all functions $f_1(\bx^1), \ldots, f_s(\bx^s)$
(or their negations) regarded as projections of $f$,
the whole argument can be repeated recursively.
In the end, one sees that $f$ and $g$ can be expressed by
the same formula, and so $f = g$. This concludes the proof.
\end{proof}

\begin{corollary*}
    For any read-once function $f$ over $\bl$ depending on $n$ variables
    it holds that
    \begin{equation*}
        T_{\bl}(f) \le 2^l \cdot \binom{n}{l} = O(n^l).
    \end{equation*}
\end{corollary*}

\section{Discussion}
\label{s:conc}

It is interesting to note that our result gives a stronger form of Kuznetsov's classic theorem on
read-once representations~\cite{kuzn}. The original result can be reformulated
as follows: for any given Boolean function $f$ and any two trees
$T_1$ and $T_2$ representing $f$, there exists a one-to-one correspondence $\phi$
between the sets of internal nodes of $T_1$ and $T_2$ such that the functions
represented by each pair of matching nodes are either equal or each other's negations.
This fact was independently proved by Aaronson~\cite{scaa}, who also developed
an $O(N^{\log_2 3} \log N)$ algorithm for transforming the truth table of $f$ into
such a tree. Note that a sequence of $O(N)$-sized circuits that check the existence of
and output read-once representations over $\bl$ for any fixed $l$ was constructed
in~\cite{aavpmi23}. In these results $N = 2^n$ is the input length.

Now suppose $T_1$ and $T_2$ are trees
representing $n$-variable Boolean functions $f_1$ and $f_2$, and it is known
a priori that these trees do not contain nodes labeled with prime functions of
arity greater than $l$.
Our technique reveals that in order to prove
the existence of a correspondence $\phi$
it is sufficient to verify that $f_1$ and $f_2$ agree on an $O(n^l)$-sized
set of input vectors. While Kuznetsov's theorem does not concern itself
with computational issues, our theorem shows that only a small fraction of input
vectors (in fact, a polynomial number of them, as compared to the total of $2^n$)
is needed to certify the ``similarity'' of the trees.

\subsubsection*{Acknowledgements.}

The author is indebted to Prof. Andrey~A.~Voronenko, who suggested the problem
considered in this paper.
The author also wishes to thank Maksim~A.~Bashov for useful
discussions and the anonymous referees for their advice.
This research has been supported by Russian Presidential grant MD--757.2011.9.

\end{document}